\documentclass{article}
\usepackage{ijcai23}

\usepackage{times}
\usepackage{soul}
\usepackage{url}
\usepackage[hidelinks]{hyperref}
\usepackage[utf8]{inputenc}

\usepackage[small]{caption}
\usepackage{graphicx}
\usepackage{amsmath}
\usepackage{amsthm}
\usepackage{booktabs}
\usepackage{algorithm}
\usepackage{algorithmic}
\usepackage[switch]{lineno}

\usepackage{comment}
\usepackage[hang,flushmargin]{footmisc}
\usepackage{lipsum}
\usepackage{Command}
\usepackage{multirow}

\newcommand{\argmin}{\mathop{\rm arg~min}\limits}

\usepackage{amsthm}
\theoremstyle{definition}

\newtheorem{problem}{Problem}
\newtheorem{lemma}{Lemma}

\newtheorem{definition}{Definition}
\newtheorem{thm}{Theorem}
\newtheorem{remark}{Remark}
\usepackage{xparse}

\newcommand{\supremum}{\mathop{\rm sup}\limits}
\NewDocumentCommand{\timeseries}{O{}O{}O{}}{#1_{#2}[{#3}]}
\NewDocumentCommand{\ret}{O{}O{}}{\timeseries[r][#1][#2]}
\NewDocumentCommand{\vecret}{O{}O{}}{\bm{r}_{#1}[{#2}]}
\NewDocumentCommand{\predret}{O{}O{}}{\hat{r}_{#1}[{#2}]}
\NewDocumentCommand{\signret}{O{}O{}}{\hat{b}_{#1}[{#2}]}
\NewDocumentCommand{\price}{O{}O{}}{X_{#1}[{#2}]}

\title{Doubly Robust Mean-CVaR Portfolio}

\author{
Kei Nakagawa$^1$
\and
Masaya Abe$^1$\and
Seiichi Kuroki$^{2}$\And
\affiliations
$^1$Nomura Asset Management Co,Ltd.\\
$^2$Recruit Co., Ltd.\\
\emails
kei.nak.0315@gmail.com,
masaya.abe.428@gmail.com,
seiichi\_kuroki@r.recruit.co.jp
}

\begin{document}

\maketitle

\begin{abstract}
In this study, we address the challenge of portfolio optimization, a critical aspect of managing investment risks and maximizing returns. 
The mean-CVaR portfolio is considered a promising method due to today's unstable financial market crises like the COVID-19 pandemic.
It incorporates expected returns into the CVaR, which considers the expected value of losses exceeding a specified probability level.
However, the instability associated with the input parameter changes and estimation errors can deteriorate portfolio performance.
Therefore in this study, we propose a Doubly Robust mean-CVaR Portfolio refined approach to the mean-CVaR portfolio optimization. 
Our method can solve the instability problem to simultaneously optimize the multiple levels of CVaRs and define uncertainty sets for the mean parameter to perform robust optimization.
Theoretically, the proposed method can be formulated as a second-order cone programming problem which is the same formulation as traditional mean-variance portfolio optimization.
In addition, we derive an estimation error bound of the proposed method for the finite-sample case.
Finally, experiments with benchmark and real market data show that our proposed method exhibits better performance compared to existing portfolio optimization strategies.
\end{abstract}

\section{Introduction}
\if
The Mean-Variance~(MV) optimization method has been the most widely used in practice for portfolio construction or asset allocation purposes.
The MV portfolio optimization employs expected returns and co-variances as inputs to determine the optimal weights~\cite{markowitz1952portfolio}.
However, two major issues exist for the MV portfolio optimization: (1)~the variance as a risk measure captures both upward and downward deviations in returns, and (2)~small changes in the parameters used for optimization, especially expected returns, can substantially alter the optimal weights.

With respect to (1), minimizing the variance of the portfolio also minimizes deviation above the expected return, which is not intuitive in terms of a risk measure. 
Additionally, variance by definition ignores extreme losses or tail risks.
To address the shortcoming of variance as a risk measure, Conditional Value-at-Risk~(CVaR) has been proposed as a promising alternative. 
CVaR is defined as the expected value of losses that exceed a certain probability level $\beta$. 
Furthermore, \cite{rockafellar2000optimization} showed that when the loss is defined as the negative return of a portfolio and CVaR is estimated using a finite number of past return observations, then the minimizing CVaR portfolio problem can be written as a linear programming~(LP) problem and can be efficiently solved.
Therefore, the mean-CVaR portfolio, incorporating expected returns into the minimum CVaR portfolio, is used as an alternative to the MV portfolio~\cite{agarwal2004risks,yao2013mean}.

Regarding (2), according to \cite{michaud1989markowitz}, estimating the expected return, an input to the Mean-Variance portfolio, is challenging and inherently prone to errors. 
Even a minor error in expected return can significantly impact the optimization results, leading to vastly different optimal weights. 
This has led to the criticism that the MV portfolio is an error maximizer. 
Such parameter instability can degrade the portfolio's performance~\cite{demiguel2007optimal}. 
As a result, risk-based portfolios focusing only on the risk~(covariance matrix) have gained attention in both practice and academia due to their superior performance~\cite{lee2011risk}. 
Furthermore, these risk-based portfolios are less sensitive to estimation errors in the covariance matrix~\cite{nakagawa2018risk}.
On the other hand, the optimization problem for the minimum CVaR portfolio, unlike risk-based portfolios, is formulated using a single probability level $\beta$. 
Thus, the choice of this $\beta$ can lead to vastly different portfolios, as observed in the MV portfolio~\cite{nakagawa2021rm}. 
Consequently, while CVaR exhibits favorable properties compared to other risk measures, the mean-CVaR portfolio shows parameter instability for both expected returns and the probability level $\beta$.

In this study, we introduce a Doubly Robust mean-CVaR Portfolio refined approach to the mean-CVaR portfolio optimization. 
We address the challenges stemming from both the estimation of expected returns and the selection of the CVaR probability level $\beta$.
Building upon methodologies proposed in previous literature~\cite{nakagawa2021rm,nakagawa2021carry,nakagawa2021taming}, our approach integrates information across multiple probability levels. 
This simultaneous optimization of CVaR provides a comprehensive perspective on potential downside risks, reducing the dependence on any single probability level and thereby enhancing the robustness of the resulting portfolios.

Moreover, the inherent challenge of accurately estimating expected returns, owing to the unpredictable nature of market dynamics and a myriad of influencing economic factors, cannot be understated~\cite{merton1980estimating}. 
Our proposal uniquely addresses this by implementing robust optimization techniques tailored to this context. 
We define the uncertainties associated with these returns within an elliptical uncertainty set. Drawing inspiration from multivariate statistical paradigms, this representation ensures that the resultant portfolio is not excessively sensitive to minor fluctuations in expected returns. 
Such a robust stance, as delineated by \cite{fabozzi2007robust,fabozzi2007brobust}, promises a more stable and resilient optimization framework.

Theoretically, the proposed method can be theoretically formulated as a second-order cone programming problem. This is the same formulation as MV portfolio optimization~\cite{lobo1998applications} and can be done without increasing the amount of computation.
In addition, we derive an estimation error bound of the proposed method for the finite-sample case.
Finally, experiments with benchmark and real market data show that our proposed method exhibits better performance compared to existing portfolio optimization strategies.
\fi

%Revise:
%Imortance of research 
The problem of finding the optimum portfolio for investors is known as a portfolio optimization problem. 
This has been an important research theme, both academically and practically as it is a crucial part of managing risk and maximizing returns from a set of investments.
The Mean-Variance~(MV) optimization method has historically been the cornerstone for portfolio construction and asset allocation, leveraging expected returns and covariances to ascertain optimal weights~\cite{markowitz1952portfolio,kolm201460}.

%summary of MV 
Although the MV optimization is a quantitative trade-off between portfolio return and risk~(variance), controlling the variance leads to a low deviation from the expected return with regard to both the downside and the upside. 
Additionally, variance by definition ignores extreme losses or tail risks.
In today’s increasingly volatile financial markets such as the outbreak of crises such as the COVID-19 pandemic and geopolitical war, the ability to avoid extreme drawdown can be a highly valued skill in the eyes of cautious investors~\cite{trindade2007financial}.

%summary of MCVaR 
To address the shortcoming of variance as a symmetric risk measure, an asymmetric risk measure, Conditional Value-at-Risk~(CVaR) has been proposed as a promising alternative to variance. 
CVaR is defined as the expected value of losses that exceed a certain probability level $\beta$. 
Furthermore, \cite{rockafellar2000optimization} showed that when the loss is defined as the negative return of a portfolio and CVaR is estimated using a finite number of past return observations, then the minimizing CVaR portfolio problem can be written as a linear programming~(LP) problem which can be efficiently solved.
Therefore, the mean-CVaR portfolio, incorporating expected returns into the minimum CVaR portfolio, is used as an alternative to the MV portfolio~\cite{agarwal2004risks,yao2013mean}.

%Challenge of our study
However, these portfolio optimizations have been pointed out to have serious two drawbacks of instability, which is also true for mean-CVaR portfolios:(1)~small changes in the parameters used for optimization can substantially alter the optimal weights, and (2)~estimation errors in the mean significantly impact the resulting portfolio weights.

With respect to (1), the optimization problem for the mean-CVaR portfolio is formulated using a single probability level $\beta$. 
The choice of this $\beta$ can lead to vastly different portfolios, as observed in the MV portfolio~\cite{nakagawa2021rm,nakagawa2021carry,nakagawa2021taming}. 

Regarding (2), according to \cite{michaud1989markowitz,zhang2018portfolio}, estimating the expected return, an input to the MV and mean-CVaR portfolios, is challenging and inherently prone to errors. 
Even a minor error in expected return can significantly impact the optimization results, leading to vastly different optimal weights. This has led to the criticism that the MV portfolio is an error maximizer. Such parameter instability can degrade the portfolio's performance~\cite{demiguel2007optimal,demiguel2009generalized}.

Thus, the portfolio optimization framework has to be modified when used in practice in order to achieve reliability, stability, and robustness~\cite{kolm201460}.

%Our study
Therefore, in this study, we propose a Doubly Robust mean-CVaR Portfolio refined approach to the mean-CVaR portfolio optimization. 
We address the above challenges stemming from both (1)~the selection of the CVaR probability level $\beta$ and (2)~the estimation of expected returns.
Building upon methodologies proposed in previous literature~\cite{nakagawa2021rm,nakagawa2021carry,nakagawa2021taming}, our approach integrates information across multiple probability levels. 
This simultaneous optimization of CVaR provides a comprehensive perspective on potential downside risks, reducing the dependence on any single probability level and thereby enhancing the robustness of the resulting portfolios.

Moreover, the inherent challenge of accurately estimating expected returns, owing to the unpredictable nature of market dynamics and a myriad of influencing economic factors, cannot be understated~\cite{merton1980estimating}. 
Our proposal uniquely addresses this by implementing robust optimization techniques tailored to this context. 
We define the uncertainties associated with these returns within an elliptical uncertainty set. Drawing inspiration from multivariate statistical paradigms, this representation ensures that the resultant portfolio is not excessively sensitive to minor fluctuations in expected returns. 
Such a robust stance, as delineated by \cite{fabozzi2007robust,fabozzi2007brobust}, promises a more stable and resilient optimization framework.

Theoretically, the proposed method can be formulated as a second-order cone programming problem. 
This is the same formulation as MV portfolio optimization~\cite{lobo1998applications} and can be done without increasing the amount of computation.
In addition, we derive an estimation error bound of the proposed method for the finite-sample case.
Finally, experiments with benchmark and real market data show that our proposed method exhibits better performance compared to existing portfolio optimization strategies.

In the following sections, we first formulate the VaR, CVaR, and mean-CVaR portfolio. Then, we propose the Doubly Robust mean-CVaR Portfolio 
 and prove the theoretical properties in Section 4.
In Section 5, we investigate the empirical effectiveness of our portfolio. Finally, we conclude in Section 6.

\section{Preliminary}
We first define VaR and CVaR and then use them to describe a portfolio optimization problem in this section.

Let $r_n$ be the return of asset $n~(1 \leq n \leq N)$ and $w_n$ be the portfolio weight for asset $i$. Denote $r = (r[1],..., r_N)^\top $ and $w = (w_1,..., w_N)^\top$.

Let $L(w, r)$ be a portfolio loss function defined as $L(w, r):= - w^\top r$ in this paper.
Here, we assume that $r_n$ is a random variable and has the continuous probability density function $f(r)$.
Then the probability that the loss function is less than $\alpha$ is 
\begin{equation}
    \Phi(w,\alpha) := \int_{- w^\top r \leq \alpha} f(r) dr.
\end{equation}
For simplicity, we assume that $\Phi(w,\alpha)$ is a continuous function with respect to $\alpha$.
Then we can define the VaR and CVaR with confidence level $\beta \in (0,1)$.

\begin{definition}[Value at Risk;~VaR]
\begin{equation}
    \barapp{VaR}{w}{\beta}
    := \alpha_{\beta}(w)
    = \min(\alpha:\Phi(w,\alpha)>\beta) 
\end{equation}
\label{VaR}
\end{definition}

\begin{definition}[Conditional Value at Risk;~CVaR]
\begin{align}
    \barapp{CVaR}{w}{\beta}
    & := \barapp{\phi}{w}{\beta} \\ \nonumber
    & = (1-\beta)^{-1}\int_{- w^\top r \geq \alpha_{\beta}(w)}- w^\top r f(r)dr
\end{align}
\label{CVaR}
\end{definition}

%Difference of VaR and CVaR.
By definition, CVaR is a larger and more conservative risk measure than VaR.
Another important difference between VaR and CVaR is whether it is coherent.
\cite{artzner1999coherent} proposed the coherent risk measure which characterizes the rationale of risk measure.

\begin{definition}[Coherent Risk Measure]
The risk measure $\rho$ that maps random loss $X$ to $\mathbb{R}$ is called a coherent risk measure when it satisfies the four conditions.
\begin{description}
   \item[Subadditivity:]~\\ for all random losses $X$ and $Y$ , $\rho(X + Y ) \leq \rho(X) + \rho(Y)$
   \item[Positive homogeneity:]~\\ for positive constant $a \in \mathbb{R^+}, \rho(aX) = a \rho(X)$
   \item[Monotonicity:]~\\if $X \leq Y$ for each outcome, then $\rho(X) \leq \rho(Y)$
   \item[Translation invariance:]~\\ for constant $m \in \mathbb{R}, \rho(X + m) = \rho(X) + m$
\end{description}
\label{defCoherent}
\end{definition}

CVaR is a coherent risk measure while VaR is not a coherent risk measure as it does not satisfy the Subadditivity~\cite{mcneil2015quantitative}. 
Thus CVaR is preferable to VaR from the above perspectives.

It is difficult to directly optimize the above CVaR because the integration interval depends on VaR as shown in Eq~\eqref{CVaR}.
To compute the CVaR $\phi_{\beta}(w)$ easily, we define the auxiliary function $\barapp{F}{w,\alpha}{\beta}$ as
\begin{definition}[Auxiliary Function for CVaR]
\label{AF}
\begin{equation}
    \barapp{F}{w,\alpha}{\beta}
    := \alpha + (1-\beta)^{-1}\int_{\mathbb{R}^N}\max(- w^\top r-\alpha,0)f(r)dr
    \end{equation}
\end{definition}

Then, the following relationship holds between CVaR~$\barapp{\phi}{w}{\beta}$ and its auxiliary function
~$\barapp{F}{w,\alpha}{\beta}$.
%
%Theorem
\begin{lemma}
\label{lemma1}
    For an arbitrarily fixed $w$,
    $\barapp{F}{w,\alpha}{\beta}$ is convex and continuously differentiable as a function of $\alpha$.
    The value of CVaR $\barapp{\phi}{w}{\beta}$ is given by minimizing
    $\barapp{F}{w,\alpha}{\beta}$ with respect to $\alpha$.
    \begin{equation}
        \min_{\alpha} \barapp{F}{w,\alpha}{\beta}
        =
        \barapp{\phi}{w}{\beta}
    \end{equation}
    In this formula, the set consisting of the values of $\alpha$ for which the minimum is attained, namely
    \begin{equation}
        A_{\beta}
        =
        \argmin_{\alpha} \barapp{F}{w,\alpha}{\beta}
    \end{equation}
    is a nonempty closed bounded interval.
\end{lemma}
%
%Theorem
\begin{proof}
    The proof is given in \cite{rockafellar2000optimization}.
\end{proof}

We can calculate the CVaR without obtaining VaR according to this Lemma.
If $X$ is a constraint that the portfolio must satisfy, the following Lemma holds for the portfolio optimization problem using CVaR as the risk measure.
\begin{lemma}
\label{lemma2}
    Minimizing the CVaR overall $w \in X$ is equivalent to
    minimizing $\barapp{F}{w,\alpha}{\beta}$
    overall $(w,\alpha)\in X \times \mathbb{R}$,
    in the sense that
    \begin{equation}
        \min_{w \in X} \barapp{\phi}{w}{\beta}
        =
        \min_{(w,\alpha)\in X \times \mathbb{R}}
           \barapp{F}{w,\alpha}{\beta}.
    \end{equation}
    Furthermore, if $L(w, r)$ is convex with respect to $w$,
    then $\barapp{F}{w,\alpha}{\beta}$ is convex
    with respect to $(w,\alpha)$,
    and $\barapp{\phi}{w}{\beta}$ is convex with respect to $w$.
    If $X$ is a convex set,
    the minimization problem of $\barapp{\phi}{w}{\beta}$ on $w \in X$ can be formulated
    as a convex programming problem.
\end{lemma}
\begin{proof}
    The proof is given in \cite{rockafellar2000optimization}.
\end{proof}

We approximate the function $\barapp{F}{w,\alpha}{\beta}$ by sampling a random variable $r$ from the density function $p(r)$.
When we get Q samples ${r[1], r[2], ..., r[Q]}$ by sampling or simple historical data, we can approximate the function $\barapp{F}{w,\alpha}{\beta}$ as follows.
\begin{equation}
\barapp{F}{w,\alpha}{\beta}=\alpha + (Q(1-\beta))^{-1}\sum_{q=1}^Q\max(-w^\top r[q]-\alpha,0)
\end{equation}

Finally, let $\mu = (\mu_1,.... ,\mu_N)^\top$ be the expected return for each asset and is $c$ be the required expected return. 

The mean-CVaR portfolio optimization problem can be formulated as a linear programming problem as follows:

\begin{problem}[mean-CVaR portfolio optimization]
\label{mean-CVaR}
\begin{align}
\min_{w,\alpha,u_1,...,u_Q} & \alpha + (Q(1-\beta))^{-1}\sum_{q=1}^Q u_q \label{minCVaR} \\ 
s.t.  & u_q \geq -w^\top r[q] -\alpha \quad (q=1,...,Q)\\
& u_q \geq 0 \quad (q=1,...,Q) \\
& \mu^\top w \geq c \label{nonneg}
\end{align}
\end{problem}

Here, we add constraints of portfolio weight that the sum of all the portfolio weights always equals one, and $w_n \geq 0$ indicates that investors take a long position of the $n$-th asset.

%%%%%%%%%%%%%%%%%%%%%%%%%%%%%%%%%%%%%%%%%%%%%%%%%%%%%%%%%%%%%%%%%%%%%%%%%%%%%%%%%%%%%%%%%%%%%%%%%%%%%%%%%%%%
\section{Proposed Method}
In this section, we describe doubly robust mean-CVaR portfolio optimization that solves instability in both the mean and CVaR.

We first define the deviation of CVaR.
It presents how much deviation is allowed from the original minimum CVaR value~$C_{\beta_k}$.
Here, $C_{\beta_k},k=1,...,K$ is the value of CVaR obtained by solving Problem \ref{minCVaR}.

\begin{definition}[Deviation of CVaR]
\begin{equation}
    \barapp{D}{w}{\beta_k}
    :=  \barapp{\phi}{w}{\beta_k}-C_{\beta_k} 
\label{DCVaR}
    \end{equation}
\end{definition}

Note that the equation \eqref{DCVaR} is non-negative because $C_{\beta_k}$ is minimized by $\barapp{\phi}{w}{\beta_k}$ and $\barapp{\phi}{w}{\beta_k}\geq C_{\beta_k}$ holds.

Then, the mean multiple CVaR portfolio optimization problem, in which multiple CVaRs $C_{\beta_k}$ are simultaneously optimized while minimizing negative expected return (maximize expected return) $- w^\top \mu$, is defined as follows.

\begin{problem}[mean multiple CVaR portfolio]
\label{P1}
\begin{align}
\min_{(d,w) \in \mathbb{R} \times X} &   ~~d- w^\top \mu \\
{\rm s.t.}
 &~ \barapp{D}{w}{\beta_k} \leq d\times |C_{\beta_k}| \quad (k = 1, \ldots, K)
\end{align}
\end{problem}

Here, $d$ represents the deviation percentage for $C_{\beta_k}$ optimized for each probability level $\beta_k,k=1,... , K$.

Let $\barapp{F}{w,\alpha_k}{\beta_k}$ be the auxiliary function of Definition~\ref{AF}.
Then following relationship holds between $\barapp{\phi}{w}{\beta_k}$ and $\barapp{F}{w,\alpha_k}{\beta_k}$ likewise Lemma~\ref{lemma1}.

\begin{equation}
    \barapp{\phi}{w}{\beta_k}
    =
    \min_{\alpha_k} \barapp{F}{w, \alpha_k}{\beta_k}
    \label{PhiwBeta}
\end{equation}
Using Eq. \eqref{DCVaR} and \eqref{PhiwBeta}, Problem \ref{P1} can be written as follows.
\begin{problem}
\begin{align}
\min_{(d,w) \in \mathbb{R} \times X} &   ~~d- w^\top \mu \\
    {\rm s.t.} &~
        \min_{\alpha_k} \barapp{F}{w, \alpha_k}{\beta_k}
        \leq
        d \times |C_{\beta_k}| + C_{\beta_k} \\
        &\quad \quad (k = 1, \ldots, K) \nonumber
\end{align}
\label{mCVaR_v1}
\end{problem}
Let $\alpha=(\alpha_1, \cdots, \alpha_K)^\top$ and we consider the following Problem \ref{mCVaR_v2}.
\begin{problem}
\begin{align}
\min_{(d, w,\alpha) \in \mathbb{R} \times X \times \mathbb{R}^K} &  ~~d - w^\top \mu \\
    {\rm s.t.} &~
        \barapp{F}{w, \alpha_k}{\beta_k}
        \leq
        d \times |C_{\beta_k}| + C_{\beta_k}\\
        &\quad \quad (k = 1, \ldots, K) \nonumber
\end{align}
\label{mCVaR_v2}
\end{problem}

Here, the following Lemma holds between Problem \ref{mCVaR_v1} and \ref{mCVaR_v2}.
%Theorem
\begin{lemma}~\\
\begin{enumerate}
    \item If $(d^*,w^*)$ is the optimal value for Problem \ref{mCVaR_v1},
    then $(d^*,w^*,\alpha^*)$ is the optimal value of Problem \ref{mCVaR_v2}.
    \item If $(d^{**},w^{**},\alpha^{**})$ is the optimal value for Problem \ref{mCVaR_v2},then
    $(d^{**},w^{**})$ is the optimal value for Problem \ref{mCVaR_v1}.
    %Theorem
\end{enumerate}

    \label{lemma3}
\end{lemma}
\begin{proof}
    (1) %Feasibility
    Assume that $(d^*,w^*)$ is the optimal value for Problem \ref{mCVaR_v1}.
    Because $(d^*,w^*)$ is a feasible solution of Problem \ref{mCVaR_v1}, 
    $\min_{\alpha_k} \barapp{F}{w^*, \alpha_k}{\beta_k} \leq d^* \times |C_{\beta_k}| + C_{\beta_k}$ holds for $k=1,...,K$.
    
    Define $\alpha^*=(\alpha_1^*, \ldots, \alpha_K^*)^\top$ as
    $
        \alpha_k^*
        :=
        {\rm argmin}_{\alpha_k} \barapp{F}{w^*, \alpha_k}{\beta_k}.
    $
    Then, $(d^*, w^*, \alpha^*)$ is a feasible solution of Problem \ref{mCVaR_v2}
    since $\barapp{F}{w^*, \alpha_k^*}{\beta_k} \leq d^* \times |C_{\beta_k}| + C_{\beta_k}$ holds for $k=1,...,K$.
    %Optimality
    If $(d^*, w^*, \alpha^*)$ is not the optimal solution of Problem \ref{mCVaR_v2},
    there exists a feasible solution $(\hat{d}, \hat{w}, \hat{\alpha})$ satisfying $\hat{d} - \hat{w}^{\top} \mu < d^* - w^{*\top} \mu$ and $
        {\rm min}_{\alpha_k}
            \barapp{F}{
                \hat{w}, \hat{\alpha}_k}{\beta_k
            }
        \leq
        \hat{d} \times |C_{\beta_k}| + C_{\beta_k}
        (k = 1, ..., K)
    $.
    Therefore $(\hat{d},\hat{w})$ is a feasible solution of Problem \ref{mCVaR_v1}, which contradicts that $(d^*,w*)$ is the optimal solution of Problem \ref{mCVaR_v1}.
    (2) %Feasibility
    Next, assume that $(d^{**},w^{**},\alpha^{**})$ is the optimal value for Problem \ref{mCVaR_v2}.
    Then, because $(w^{**},\alpha^{**})$ is a feasible solution of Problem \ref{mCVaR_v2},
    $
        \barapp{F}{w^{**},\alpha_k^{**}}{\beta_k}
        \leq
        d^{**} \times |C_{\beta_k}| + C_{\beta_k}
        (k = 1, ..., K)
    $ holds.
    Thus $(d^{**},w^{**})$ is a feasible solution for Problem \ref{mCVaR_v1} since
    $
        {\rm min}_{\alpha_k}
            \barapp{F}{w^{**},\alpha_k}{\beta_k}
            \leq
            \barapp{F}{w^{**},\alpha_k^{**}}{\beta_k}
            \leq
        d^{**} \times |C_{\beta_k}| + C_{\beta_k}
        (k = 1, ..., K)
    $
    holds.
    %Optimality
    if $(d^{**},w^{**})$ is not the optimal solution of
    Problem \ref{mCVaR_v1}, there exists a feasible solution $(\hat{d},\hat{w})$ satisfying $\hat{d}- \hat{w}^{\top} \mu < d^{**}- w^{**\top} \mu$.
    Define
    $
        \hat{\alpha}
        =
        (\hat{\alpha}_1,...,\hat{\alpha}_K)^\top
    $ as $
        \hat{\alpha}_k := \argmin_{\alpha_k} F_{\beta_k}(\hat{w},\alpha_k)
    $.
    Then, $(\hat{d},\hat{w},\hat{\alpha})$ is a feasible solution of Problem \ref{mCVaR_v2}, which contradicts that $d^{**}$ is the optimal solution of Problem \ref{mCVaR_v2}.
%Proof
\end{proof}
%

%%%%%%%%%%%%%%%%%%%%%%%%%%%%%%%%%%%%%%%%%%%%%%%%%%%%%%%%%%%%%%%%%%%%%%%%%%%%%%%%%%%%%%%%%%%%%%%%%%%%%%%%%%%%%%%
%
According to Lemma \ref{lemma3}, Problem \ref{P1} and \ref{mCVaR_v2} are a equivalent problem.
When ${r[1], \ldots, r[Q]}$ are obtained by sampling, 
the function $\barapp{F}{w, \alpha_k}{\beta_k}$ is approximated as follows.
\begin{equation}
    \barapp{F}{w, \alpha_k}{\beta_k}
    \simeq
    \alpha_k + \frac{1}{Q \paren{1 - \beta_k}}
    \sum_{q=1}^Q
    [ -w^\top r[q] - \alpha_k]^+
\end{equation}

Now that we can formulate the mean-multiple CVaR portfolio optimization problem, we assume the following ellipsoidal uncertainty set $U_E$ for expected returns~\cite{fabozzi2007robust}.

$$U_E:=U_{\delta,\Sigma_{\hat{\mu}}}(\hat{\mu})=\{\mu \in \mathbb{R}^N |~ (\mu -\hat{\mu})^\top \Sigma_{\hat{\mu}}^{-1}(\mu -\hat{\mu})  < \delta^2 \}$$

Then, we can formulate our doubly robust mean-CVaR portfolio optimization problem as a solution to the following second-order cone programming problem.

\begin{thm}
The doubly robust mean-CVaR portfolio optimization problem can be formulated as a solution to the following second-order cone programming problem.
\begin{align*}
        \underset{
            d, w, \alpha, t
        }{
            {\rm min}
        }
        & \; d -\sum_{n=1}^N w_n \mu_n + \delta \sqrt{w^\top \Sigma_{\hat{\mu}} w} \\
        {\rm s.t.}
        & \; t_{qk} \geq 0 \\
        & \; t_{qk} \geq - w^\top r[q] - \alpha_k \\
        & \;
            \alpha_k + \frac{1}{Q \paren{1 - \beta_k}}
            \sum_{q=1}^Q t_{qk}
            \leq
%            (1+d)C_{\beta_k}
            d\times |C_{\beta_k}| + C_{\beta_k}
    \end{align*}
\label{thm}
\end{thm}
\begin{proof}
        Since $\sqrt{w^\top \Sigma_{\hat{\mu}} w}$ is a second-order cone term and the rest are all linear expressions, this is a second-order cone programming problem.
\end{proof}

\begin{remark}

Assume the following rectangular uncertainty set $U_R$ for expected returns
$$U_R:=U_{\delta}(\hat{\mu})=\{\mu \in \mathbb{R}^N |~ |\mu -\hat{\mu}| < \delta \}$$

We can write the doubly robust mean-CVaR portfolio optimization problem as a linear programming problem similar to the original mean-CVaR optimization problem~(Problem \ref{mean-CVaR}).
\end{remark}

\section{Theoretical Analysis}
In this section, we conducted a theoretical analysis of the excess risk associated with the proposed algorithm.

Here, we consider an optimization problem where the feasible region is a convex set $S$ and denote the empirical risk minimizer and the the true risk minimizer as $\hat{w} \in \argmin_{w \in S} \widehat{L}(w) = \frac{1}{Q} and \sum_{q=1}^{Q} L_q(w, r[q])$ and $w^{*} = \argmin_{w \in S}L(w,r)$, respectively.

The aim of this section is to demonstrate the relationship between $L(\hat{w}, r)$ and $L(w^{*}, r)$  for the finite-sample case.

In our preliminary analysis, we show that problem~\ref{mCVaR_v2} is equivalent to problem~\ref{equivalent_prob}, which will be presented later in this paper.
\begin{lemma}
Assuming that $(d^*,w^*, \alpha^*)$ is the optimal solution to problem~\ref{mCVaR_v2}, then $w^*$ is the optimal solution to the following problem.
\begin{problem}
\begin{align}
\min_{w \in \mathbb{R}^K} - w^\top \mu \\
    {\rm s.t.} &~
        \barapp{F}{w, \alpha^*}{\beta_k}
        \leq
        d \times |C_{\beta_k}| + C_{\beta_k}\\
        & \quad (k = 1, \ldots, K) \nonumber.
\end{align}
\label{equivalent_prob}
\end{problem} 
\begin{proof}
Similar to the proof of Lemma~\ref{lemma3}, given that all objective functions and constraints are convex in both problem~\ref{mCVaR_v2} and problem~\ref{equivalent_prob}, their equivalence can be proven by comparing the KKT conditions.
\end{proof}
\end{lemma}
As a remark, a theorem analogous to the aforementioned also holds for the second-order cone programming problem as formulated in Theorem~\ref{thm}.

%TODO: fix the sentence
To derive the bound of the excess risk $L(\hat{w}, r) - L(w^{*}, r)$, we use the Rademacher complexity, which captures the complexity of a set of functions by measuring the capability of a hypothesis class to correlate with the random noise.
\begin{definition}[Rademacher complexity~\cite{bartlett2002rademacher}]
Let $\mathcal{H}$ be a set of real-valued functions defined over a set $\mathcal{X}$. 
Given a sample $(x_1,\dots,x_m) \in \mathcal{X}^m$ independently and identically drawn from a distribution $\mu$, the Rademacher complexity of $\mathcal{H}$ is defined as
\begin{equation*}
\mathfrak{R}_{\mu,m}(\mathcal{H}) = \mathbb{E}_{x_1,\dots,x_m}\mathbb{E}_{\sigma}\left[ \supremum_{h \in \mathcal{H}} \frac{1}{m}\sum_{i=1}^{m} \sigma_{i} h(x_i) \right],
\end{equation*}
where the inner expectation is taken over $\sigma = (\sigma_1, \dots, \sigma_{m})$ which are mutually independent uniform random variables taking values in $\{+1, -1\}$.
\end{definition}

Finally, we derive the following theorem which demonstrates the excess risk bound for the finite-sample case.
\begin{thm}
\label{thm:finite-sample-bound-general-setting}
  Let ($r[1], \dots, r[Q]$) be a set of examples independently and identically drawn from a distribution with expected value $r$, and let $L_1(w), \dots, L_Q(w)$ independently follow a sub-Gaussian distribution with parameter $\sigma$ and expected value $L(w, \mu)$.
  
  Given that $||w||_1 = 1$ and $\mathbb{E}[||r[q]||_2^2] \leq \kappa^2$, with a probability of at least $1-\delta$,
  \begin{equation*}
    L(\hat{w}, r) - L(w^{*}, r) \leq 4 \frac{\kappa}{\sqrt{Q}}+\frac{\sqrt{2\sigma^2}}{\sqrt{Q}}\sqrt{\log\frac{1}{\delta}}.
  \end{equation*}
\end{thm}
\begin{proof}
  % Note:sub-gaussianの仮定を外すと損失関数にboundが必要になる
  % 楕円形制約のrademacher compはここで証明されている？https://users.cs.duke.edu/~cynthia/docs/TulabandhulaRuISAIM14Structure.pdf
Let $H(r[1], \dots, r[Q]) = \supremum_{w \in S} |L(w, \mu) - \widehat{L}({w})|$, where a random sample $(r[1], \dots, r[Q])$ is independent and identically drawn from distribution with expected value $r$.
Consequently, applying the Hoeffding bound for sub-Gaussian random variables, we can derive the following inequality, with probability greater than 1 - $\delta$% need citation: (Wainwright 2019, Proposition 2.5)
\begin{align*}
&H(r[1], \dots, r[Q])\\
&\leq \mathbb{E}_r[H(r[1], \dots, r[Q])] + \frac{\sqrt{2\sigma^2}}{\sqrt{Q}}\sqrt{\log\frac{1}{\delta}}\\
& \leq 4\mathfrak{R}_{r, Q}(\mathcal{F}) +\frac{\sqrt{2\sigma^2}}{\sqrt{Q}}\sqrt{\log\frac{1}{\delta}}.
\end{align*}
Then, 
\begin{align*}
\mathfrak{R}_{r, Q}(\mathcal{F}) 
&= \mathbb{E}[\supremum_{f \in \mathcal{F}} \frac{1}{Q} \sum_{q=1}^Q \sigma_q f(r[q])]\\
&= \mathbb{E}[\supremum_{w \in S} \frac{1}{Q} \sum_{q=1}^Q -w (\sigma_q r[q])]\\
&\leq \mathbb{E}[\supremum_{w \in S} \frac{1}{Q}\sum_{q=1}^Q ||w||_2 ||r[q]||_2]\\
&\leq \frac{1}{Q}\mathbb{E}\left[\sum_{q=1}^Q \sqrt{||r[q]||_2^2}\right]\\
&\leq \frac{1}{Q}\sqrt{\mathbb{E}\left[\sum_{q=1}^Q ||r[q]||_2^2\right]}\leq \frac{\kappa}{\sqrt{Q}}.
\end{align*}
Overall, combining these inequalities proves the theorem.
\end{proof}
Here, note that the normal distribution $\mathcal{N}(\mu, \sigma^2)$ is a sub-Gaussian distribution with parameter $\sigma$. Furthermore, when the returns of each asset follow a normal distribution, the objective function represented by their linear combination also follows a normal distribution.

\section{Experiment}
In this section, we report the results of our empirical studies with well-known benchmark and real-world datasets.
First, we evaluate the robustness of several portfolio strategies including our proposed strategy. Next, we compare the out-of-sample performance among them.

\subsection{Dataset}
In the experiments, we use well-known academic benchmarks called Fama and French (FF) datasets \cite{fama1992cross} to ensure the reproducibility of the experiment.
This FF dataset is public and is readily available to anyone\footnote{\url{http://mba.tuck.dartmouth.edu/pages/faculty/ken.french/data_library.html}}.
The FF datasets have been recognized as standard datasets and heavily adopted in finance research because they extensively cover asset classes and very long historical data series.
We apply the FF48 dataset which contains monthly returns of 48 portfolios representing different industrial sectors.
Next, we prepare a stock dataset corresponding to Morgan Stanley Capital International~(MSCI) U.S. index.
The MSCI indices comprise the large and mid-cap segments of the U.S. markets and are widely used as a benchmark for institutional investors.
We use FF48 datasets from January 1981 to December 2020 and MSCI U.S. datasets from January 2001 to December 2020 for out-of-sample periods.

\subsection{Experimental Settings}
In our empirical studies, the tested portfolio models have the following meanings:
\begin{itemize}
    \item "EW" stands for equally-weighted~(EW) portfolio~\cite{demiguel2007optimal}.
    The EW outperforms its market capitalization index over the long-term periods.
    \item "mean-CVaR" stands for mean-CVaR portfolio with single $\beta$ in Problem~\ref{mean-CVaR}. 
    We implement five patterns of $\beta$ = \{0.95, 0.96, 0.97, 0.98, 0.99\} and use the latest 5 years' daily data to calculate CVaR and the expected return of each stock.
    We set $c$ as the average of the expected return of each stock.
    \item "mean-MCVaR" stands for mean-multiple CVaR~\cite{nakagawa2020RM} with the same constraints of Problem~\ref{mean-CVaR}. We implement five patterns of $\beta_k$ = \{0.95, 0.96, 0.97, 0.98, 0.99\} to calculate $C_{\beta_k}$. 
    We use the latest 5 years' daily data to calculate CVaR and the expected return of each stock and set $c$ as the average of the expected return of each stock.
    \item "DR-MCVaR" stands for our proposed portfolio in Theorem~\ref{thm}. 
    We implement five patterns of $\beta_k$ = \{0.95, 0.96, 0.97, 0.98, 0.99\} to calculate $C_{\beta_k}$. We use the latest 5 years' daily data to calculate CVaR, the expected return of each stock, and $\Sigma_{\hat{\mu}}$. For the parameter $\delta$, we use the values from 99, 95, and 90\% confidence interval.
\end{itemize}

Each portfolio is updated by sliding one month ahead from the out-of-sample period. 

\begin{table*}[t]
\centering
\caption{Robustness measure among different portfolio strategies}
\label{tbl:robust}
\scalebox{0.85}{
\begin{tabular}{llrrrrrlrrrr}\hline
\multirow{2}{*}{} & \multirow{2}{*}{EW}         & \multicolumn{5}{l}{mean-CVaR}                                                                                                        & \multirow{2}{*}{\begin{tabular}[c]{@{}l@{}}mean-\\ MCVaR\end{tabular}} & \multicolumn{4}{l}{DR-MCVaR}                                                                            \\ 
                  &                             & \multicolumn{1}{l}{95\%} & \multicolumn{1}{l}{96\%} & \multicolumn{1}{l}{97\%} & \multicolumn{1}{l}{98\%} & \multicolumn{1}{l}{99\%} &                                                                        & \multicolumn{1}{l}{0\%} & \multicolumn{1}{l}{1\%} & \multicolumn{1}{l}{5\%} & \multicolumn{1}{l}{10\%}  \\ \hline
FF48$\downarrow$              & \multicolumn{1}{r}{18.04\%} & 51.30\%                  & 48.33\%                  & 53.23\%                  & 52.52\%                  & 60.82\%                  & \multicolumn{1}{r}{57.24\%}                                            & 28.18\%                 & 26.95\%                 & 26.85\%                 & 26.93\%                   \\
MSCI U.S.$\downarrow$         & \multicolumn{1}{r}{43.62\%} & 219.12\%                 & 218.33\%                 & 236.79\%                 & 234.28\%                 & 245.67\%                 & \multicolumn{1}{r}{209.42\%}                                           & 168.67\%                & 154.34\%                & 154.10\%                & 153.98\% \\ \hline
\end{tabular}
}
\end{table*}

\subsection{Performance Measures}

Hereafter, let $r_{it}$ be the realized return of $i$ asset at time $t$ and $w_{it}$ be the weight of $i$ asset in the portfolio at time $t$. 
Denote the $r_{t} =(r_{it},...,r_{nt})^\top$, $w_{t}=(w_{it},...,w_{nt})^\top$.

\subsubsection{Robustness}
First, we define turnover~(TO), a performance measure that represents the robustness of a portfolio.
The TO indicates the volumes of rebalancing. 
A portfolio can be seen as more robust if it is associated with a low TO~\cite{fabozzi2007robust}.
This is because the TO represents the sensitivity to parameter changes due to portfolio rebalancing, and since a high TO inevitably generates high explicit and implicit trading costs, the portfolio return is reduced. 
The one-way annualized TO is calculated as an average absolute value of the rebalancing trades over all the trading periods:
\begin{align}
    {\bf TO} = \frac{12}{2(T-1)}\sum_{t=1}^{T-1} ||w_t-w_{t-1}^{-}||_1
\end{align}
where $T-1$ indicates the total number of the rebalancing periods and $w_{t-1}^{-}$ is the re-normalized portfolio weight vector before rebalance.

\begin{equation}
w_{t-1}^{-} = \frac{
    w_{t-1} \otimes \paren{1 + r_{t}}
}{
    1 + w_{t-1}^{\top} r_{t}
}
\end{equation}

where the operator $\otimes$ denotes the Hadamard product.

\begin{table*}[t]
\centering
\caption{Profitability and risk measures among different portfolio strategies}
\label{tbl:performance}
\scalebox{0.85}{
\begin{tabular}{lrrrrrrrrrrr} \hline
\multicolumn{1}{c}{\multirow{2}{*}{}} & \multicolumn{1}{c}{\multirow{2}{*}{EW}} & \multicolumn{5}{c}{mean-CVaR}                                                                                                        & \multicolumn{1}{c}{\multirow{2}{*}{\begin{tabular}[c]{@{}c@{}}mean-\\ MCVaR\end{tabular}}} & \multicolumn{4}{c}{DR-MCVaR} \\
\multicolumn{1}{c}{}                  & \multicolumn{1}{c}{}                    & \multicolumn{1}{c}{95\%} & \multicolumn{1}{c}{96\%} & \multicolumn{1}{c}{97\%} & \multicolumn{1}{c}{98\%} & \multicolumn{1}{c}{99\%} & \multicolumn{1}{c}{}                                                                       & \multicolumn{1}{c}{0\%} & \multicolumn{1}{c}{1\%} & \multicolumn{1}{c}{5\%} & \multicolumn{1}{c}{10\%}  \\ \hline
\multicolumn{12}{c}{Panel A: FF48} \\ \hline
AR$\uparrow$                                    & 11.95\%                                 & 13.06\%                  & 13.18\%                  & 13.45\%                  & 13.79\%                  & 13.24\%                  & 13.82\%                                                                                    & 14.92\%                 & 14.45\%                 & 14.41\%                 & 14.39\%                   \\
RISK$\downarrow$                                  & 19.21\%                                 & 12.39\%                  & 12.53\%                  & 12.74\%                  & 13.28\%                  & 13.49\%                  & 12.82\%                                                                                    & 12.18\%                 & 11.84\%                 & 11.85\%                 & 11.86\%                   \\
R/R$\uparrow$                                   & 0.62                                    & 1.05                     & 1.05                     & 1.06                     & 1.04                     & 0.98                     & 1.08                                                                                       & 1.22                    & 1.22                    & 1.22                    & 1.21                      \\
MaxDD$\downarrow$                                 & -59.92\%                                & -45.31\%                 & -44.35\%                 & -44.30\%                 & -46.98\%                 & -51.22\%                 & -47.34\%                                                                                   & -47.18\%                & -47.18\%                & -47.40\%                & -47.53\%                  \\
CR$\uparrow$                                    & 0.20                                     & 0.29                     & 0.30                      & 0.30                      & 0.29                     & 0.26                     & 0.29                                                                                       & 0.32                    & 0.31                    & 0.30                     & 0.30                       \\ \hline
\multicolumn{12}{c}{Panel B: MSCI U.S.} \\ \hline
AR$\uparrow$                                    & 8.60\%                                  & 7.99\%                   & 8.51\%                   & 8.29\%                   & 8.82\%                   & 9.99\%                   & 8.54\%                                                                                     & 9.47\%                  & 9.48\%                  & 9.48\%                  & 9.48\%                    \\
RISK$\downarrow$                                  & 17.94\%                                 & 11.32\%                  & 11.37\%                  & 11.54\%                  & 11.82\%                  & 11.81\%                  & 11.25\%                                                                                    & 11.09\%                 & 11.07\%                 & 11.07\%                 & 11.07\%                   \\
R/R$\uparrow$                                   & 0.48                                    & 0.71                     & 0.75                     & 0.72                     & 0.75                     & 0.85                     & 0.76                                                                                       & 0.85                    & 0.86                    & 0.86                    & 0.86                      \\
MaxDD$\downarrow$                                 & -55.65\%                                & -36.07\%                 & -34.85\%                 & -36.81\%                 & -34.26\%                 & -35.74\%                 & -36.30\%                                                                                   & -38.85\%                & -39.47\%                & -39.47\%                & -39.46\%                  \\
CR$\uparrow$                                    & 0.15                                    & 0.22                     & 0.24                     & 0.23                     & 0.26                     & 0.28                     & 0.24                                                                                       & 0.24                    & 0.24                    & 0.24                    & 0.24 \\ \hline  
\end{tabular}
}
\end{table*}

\subsubsection{Profitability and Risk}
Next, we compare the out-of-sample profitability of the portfolios.
In evaluating the portfolio strategy, we use the following measures that are widely used in the field of finance~\cite{brandt2010portfolio,nakagawa2020ric}.

The (realized) portfolio return at time $t$ is defined as
\begin{equation}
    R_t = \sum_{i=1}^n r_{it}w_{it-1}.
\end{equation}
 
We evaluate the portfolio strategy using its annualized return~(AR), the risk as the standard deviation of return~(RISK), and the risk/return~(R/R) as the portfolio strategy. 
R/R is a risk-adjusted return measure for a portfolio strategy.
\begin{align}
    {\bf AR} &= \prod_{t=1}^T (1+R_t)^{12/T}-1 \\
    {\bf RISK} &= \sqrt{\frac{12}{T-1}\times(R_t-\mu)^2}\\
    {\bf R/R} &= {\bf AR}/{\bf RISK}
\end{align}

Here, $\mu = (1/T) \sum_{t=1}^T R_t$ is the average return of the portfolio.

We evaluate the maximum drawdown~(MaxDD), which is another widely used risk measure \cite{magdon2004maximum,shen2017portfolio}, for the portfolio strategy.
In particular, MaxDD is defined as the largest drop from an extremum:
\begin{align}
    {\bf MaxDD} &= \min_{k \in [1,T]}\left(0,\frac{W_k}{\max_{j \in [1,k]} W_{j}}-1\right) \\
    W_k &= \prod_{l=1}^k (1+R_l).
\end{align}
where $W_k$ is the cumulative return of the portfolio until time $k$.

We also use the Calmar ratio \cite{young1991calmar}, another definition of adjuster returns. 
The Calmar ratio is defined as 
$${\bf CR} := \bf{AR} / {\bf MaxDD}$$. 

Note that while both R/R and CR are adjusted returns by risk measures, CR is more sensitive to drawdown events that occur less frequently (e.g., financial crises).

\subsection{Results}
\subsubsection{Robustness}
Table~\ref{tbl:robust} presents the TO, the robustness measures of the portfolio for the FF48 and MSCI U.S. datasets across portfolios.

The EW portfolio has a considerably lower TO compared to other portfolios. 
For the FF48 and MSCI U.S. datasets, the EW turnover ratios are 18.04\% and 43.62\%, respectively. 
The lower TO in the EW portfolio can be attributed to its strategy of holding equal weights of all assets, which generally requires less frequent rebalancing compared to portfolios that are managed more actively.

When we compare our proposed portfolio~(DR-MCVaR) portfolio with others, it exhibits a notably lower TO than the mean-CVaR and mean-MCVaR portfolios. 
The TOs for DR-MCVaR range between 26.85\% and 28.18\% for FF48 and between 153.98\% and 168.67\% for MSCI U.S.
The relatively consistent TOs across different parameters within the DR-MCVaR portfolio indicate robustness in its strategy, which might be beneficial in maintaining a steady portfolio performance across varying market conditions.

\subsubsection{Profitability and risk}
Table~\ref{tbl:performance} presents the AR, RISK, R/R, MaxDD, and CR for the FF48 and MSCI U.S. datasets across portfolios.

First, we focus on the profitability measures. In the FF48 panel, the DR-MCVaR portfolio consistently outperforms other strategies across all parameters, with AR ranging from 14.39\% to 14.92\%. 
In the MSCI U.S. panel, it also exhibits strong performance, with ARs around 9.47\% to 9.48\%, which are higher than the EW and mean-MCVaR but slightly lower than the highest mean-CVaR at 99\% confidence level.
The DR-MCVaR portfolio showcases superior R/R ratios across all parameters in both panels, indicating a more favorable return per unit of risk. In the FF48 panel, the R/R ratios range from 1.21 to 1.22, while in the MSCI U.S. panel, they are around 0.85 to 0.86.
In terms of the CR, the DR-MCVaR portfolio demonstrates the highest values in the FF48 panel, ranging from 0.30 to 0.32, indicating better performance relative to its maximum drawdown. 
In the MSCI U.S. panel, it maintains a consistent CR of 0.24, which is competitive but not the highest among the strategies.

Then, we move on to the risk measures. The DR-MCVaR portfolio exhibits the lowest annualized risk across all parameters in both panels, indicating a lower risk profile compared to other strategies. 
In terms of maximum drawdown, the DR-MCVaR portfolio shows moderate MaxDD values, which are not the lowest among the strategies but still represent a relatively lower risk of substantial losses.

The DR-MCVaR portfolio generally exhibits superior performance in terms of profitability measures, especially in the FF48 panel where it consistently outperforms other portfolios in AR, R/R, and CR. This suggests that the DR-MCVaR portfolio, as proposed, offers a promising strategy for achieving higher returns while efficiently managing risk, as evidenced by its high R/R and CR values.

The DR-MCVaR portfolio, which is the focus of our study, consistently exhibits strong performance across most metrics in both panels, showcasing both high returns and favorable risk-adjusted returns. Its superior performance in the R/R metric indicates a robust risk management strategy, achieving high returns while mitigating risks effectively.

\section{Conclusion}
Our study makes the following contributions:
\begin{itemize}
    \item We propose The DR-MCVaR method effectively addresses the two major shortcomings associated with the MV method: the inadequate representation of risk through variance and the sensitivity to parameter estimations, particularly expected returns.
    \item The theoretical formulation of our method as a second-order cone programming problem aligns with the MV portfolio optimization, maintaining computational efficiency without increasing the computational burden. In addition, we derive an estimation error bound of the proposed method for the finite-sample case.
    \item We demonstrate that the empirical evidence from experiments with benchmark and real market data substantiates the efficacy of our proposed method, showcasing superior performance compared to existing portfolio optimization strategies.
\end{itemize}

\end{document}